\documentclass[12pt]{article}

\usepackage{amsthm,amsmath,amssymb}

\usepackage{graphicx}

\usepackage[colorlinks=true,citecolor=black,linkcolor=black,urlcolor=blue]{hyperref}

\theoremstyle{plain}
\newtheorem{theorem}{Theorem}
\newtheorem{lemma}[theorem]{Lemma}
\newtheorem{corollary}[theorem]{Corollary}

\theoremstyle{definition}
\newtheorem{definition}[theorem]{Definition}

\newtheorem{conjecture}[theorem]{Conjecture}

\theoremstyle{remark}

\title{\bf Bounding sequence extremal functions\\ with formations}

\author{J.T. Geneson\thanks{Supported by the NSF Graduate Research Fellowship under Grant No. 1122374.} \qquad Rohil Prasad \qquad  Jonathan Tidor\\
\small Department of Mathematics, MIT\\[-0.8ex]
\small Massachusetts, U.S.A.\\
\small\tt geneson@math.mit.edu\\[-0.8ex]
\small\tt prasad01@college.harvard.edu\\[-0.8ex]
\small\tt jtidor@mit.edu
}
\date{}

\begin{document}

\maketitle


\begin{abstract}
An $(r, s)$-formation is a concatenation of $s$ permutations of $r$ letters. If $u$ is a sequence with $r$ distinct letters, then let $\mathit{Ex}(u, n)$ be the maximum length of any $r$-sparse sequence with $n$ distinct letters which has no subsequence isomorphic to $u$. For every sequence $u$ define $\mathit{fw}(u)$, the formation width of $u$, to be the minimum $s$ for which there exists $r$ such that there is a subsequence isomorphic to $u$ in every $(r, s)$-formation. We use $\mathit{fw}(u)$ to prove upper bounds on $\mathit{Ex}(u, n)$ for sequences $u$ such that $u$ contains an alternation with the same formation width as $u$. 

We generalize Nivasch's bounds on $\mathit{Ex}((ab)^{t}, n)$ by showing that $\mathit{fw}((12 \ldots l)^{t})=2t-1$ and $\mathit{Ex}((12\ldots l)^{t}, n) =n2^{\frac{1}{(t-2)!}\alpha(n)^{t-2}\pm O(\alpha(n)^{t-3})}$ for every $l \geq 2$ and $t\geq 3$, such that $\alpha(n)$ denotes the inverse Ackermann function. Upper bounds on $\mathit{Ex}((12 \ldots l)^{t} , n)$ have been used in other papers to bound the maximum number of edges in $k$-quasiplanar graphs on $n$ vertices with no pair of edges intersecting in more than $O(1)$ points.

If $u$ is any sequence of the form $a v a v' a$ such that $a$ is a letter, $v$ is a nonempty sequence excluding $a$ with no repeated letters and $v'$ is obtained from $v$ by only moving the first letter of $v$ to another place in $v$, then we show that $\mathit{fw}(u)=4$ and $\mathit{Ex}(u, n) =\Theta(n\alpha(n))$. Furthermore we prove that $\mathit{fw}(abc(acb)^{t})=2t+1$ and $\mathit{Ex}(abc(acb)^{t}, n) = n2^{\frac{1}{(t-1)!}\alpha(n)^{t-1}\pm O(\alpha(n)^{t-2})}$ for every $t\geq 2$.


\bigskip\end{abstract}

\section{Introduction}

A {\it Davenport-Schinzel} sequence of order $s$ is a sequence with no adjacent repeated letters which has no alternating subsequence of length $s+2$. Upper bounds on the lengths of Davenport-Schinzel sequences provide bounds on the complexity of lower envelopes of solution sets to linear homogeneous differential equations of limited order \cite{3} and on the complexity of faces in arrangements of arcs with a limited number of crossings \cite{1}.

A sequence $s$ {\it contains} a sequence $u$ if some subsequence of $s$ can be changed into $u$ by a one-to-one renaming of its letters. If $s$ does not contain $u$, then $s$ {\it avoids} $u$. The sequence $s$ is called {\it r-sparse} if any $r$ consecutive letters in $s$ are pairwise different. If $u$ is a sequence with $r$ distinct letters, then the extremal function $\mathit{Ex}(u, n)$ is the maximum length of any $r$-sparse sequence with $n$ distinct letters which avoids $u$.

A {\it generalized Davenport-Schinzel sequence} is an $r$-sparse sequence with no subsequence isomorphic to a fixed forbidden sequence with $r$ distinct letters. Fox {\it et al.} \cite{5} and Suk {\it et al.} \cite{10} used bounds on the lengths of generalized Davenport-Schinzel sequences to prove that $k$-quasiplanar graphs on $n$ vertices with no pair of edges intersecting in more than $t$ points have at most $(n\log n)2^{\alpha(n)^{c}}$ edges, where $\alpha(n)$ denotes the inverse Ackermann function and $c$ is a constant that depends only on $k$ and $t$.

If $a$ and $b$ are single letters, then $\mathit{Ex}(a, n)=0, \mathit{Ex}(a b, n)=1, \mathit{Ex}(a b a, n)= n$ and $\mathit{Ex}(a b a b, n)=2n-1$. Nivasch \cite{8} and Klazar \cite{7} determined that $ \mathit{Ex}(ababa,n) \sim 2n\alpha(n)$. Agarwal, Sharir, and Shor \cite{2} proved the lower bound and Nivasch \cite{8} proved the upper bound to show that if $u$ is an alternation of length $2t+4$, then $\mathit{Ex}(u, n)=n2^{\frac{1}{t!}\alpha(n)^{t}\pm O(\alpha(n)^{t-1})}$ for $t\geq 1$.

If $u$ is a sequence with $r$ distinct letters and $c\geq r$, then let $\mathit{Ex_{c}}(u, n)$ be the maximum length of any $c$-sparse sequence with $n$ distinct letters which avoids $u$. Klazar \cite{6} showed that $\mathit{Ex_{c}}(u, n)=\Theta(\mathit{Ex_{d}}(u, n))$ for all fixed $c, d\geq r$.

\begin{lemma} \cite{6} \label{1.1} If $u$ is a sequence with $r$ distinct letters, then $\mathit{Ex_{d}}(u, n)\leq \mathit{Ex_{c}}(u, n)\leq(1+\mathit{Ex_{c}}(u, d-1))\mathit{Ex_{d}}(u, n)$ for all $n\geq 1$ and $d\geq c\geq r$. \end{lemma}

An $(r, s)$-{\it formation} is a concatenation of $s$ permutations of $r$ distinct letters. For example {\it abcddcbaadbc} is a $(4,3)$-formation.

\begin{definition} $F_{r,s}(n)$ is the maximum length of any $r$-sparse sequence with $n$ distinct letters that avoids every $(r, s)$-formation.\end{definition}

Klazar \cite{6} proved that $F_{r,2}(n)=O(n)$ and $F_{r,3}(n)= O(n)$ for every $r$. Nivasch \cite{8} proved that $F_{r,4}(n)=\Theta(n\alpha(n))$ for $r\geq 2$. Agarwal, Sharir, and Shor \cite{2} proved the lower bound and Nivasch \cite{8} proved the upper bound to show that $F_{r,s}(n)=n2^{\frac{1}{t!}\alpha(n)^{t}\pm O(\alpha(n)^{t-1})}$ for all $r\geq 2$ and odd $s \geq 5$ with $t= \frac{s-3}{2}$.

Nivasch \cite{8} proved that $\mathit{Ex}(u, n)\leq F_{r,s-r+1}(n)$ for any sequence $u$ with $r$ distinct letters and length $s$ by showing that every $(r, s -r+1)$ formation contains $u$. 

\begin{definition} The {\it formation width} of $u$, denoted by $\mathit{fw}(u)$, is the minimum value of $s$ such that there exists an $r$ for which every $(r, s)$-formation contains $u$. The {\it formation length} of $u$, denoted by $\mathit{fl}(u)$, is the minimum value of $r$ such that every $(r, \mathit{fw}(u))$-formation contains $u$. \end{definition}

By Nivasch's proof, $\mathit{fw}(u)\leq s -r+1$ for every sequence $u$ with $r$ distinct letters and length $s$. The next two facts follow from the definition of $\mathit{fw}$.

\begin{lemma} \label{1.2} If $u$ contains $v$, then $\mathit{fw}(v)\leq \mathit{fw}(u)$. \end{lemma}

\begin{lemma} \label{1.3} If $u$ begins with the letter $a$, then $\mathit{fw}(a u)=\mathit{fw}(u)+1$. \end{lemma}

Lemma \ref{1.1} implies that $\mathit{fw}(u)$ and $\mathit{fl}(u)$ can be used to obtain upper bounds on $\mathit{Ex}(u, n)$.

\begin{lemma} \label{1.4} For any sequence $u$ with $r$ distinct letters and fixed $c$ with $c\geq r, \mathit{Ex_{c}}(u, n)=O(F_{\mathit{fl}(u),\mathit{fw}(u)} (n))$. \end{lemma}

In this paper we use $\mathit{fw}(u)$ primarily in order to prove tight upper bounds on $\mathit{Ex}(u, n)$ for several classes of sequences $u$ such that $u$ contains an alternation with the same formation width as $u$. We also bound and evaluate $\mathit{fw}$ for various other families of sequences in order to develop a classification of all sequences in terms of their formation widths.

If $\alpha_{t}$ is an alternation of length $t$ for $t\geq 2$, then $\mathit{fw}(\alpha_{t})\leq t-1$ since every $(r, t-1)$ formation contains $\alpha_{t}$ for $r\geq 2$. Any $(r, t-2)$-formation in which order of letters reverses in adjacent permutations avoids $\alpha_{t}$, so $\mathit{fw}(\alpha_{t})=t-1$. Pettie \cite{9} used the fact that every $(4, 4)$-formation contains {\it abcacbc} to prove the upper bound $\mathit{Ex}(abcacbc, n)=O(n\alpha(n))$. Since any $(r, 3)$ formation with order reversing in adjacent permutations would avoid {\it abcacbc}, then $\mathit{fw}(abcacbc)=4$. Similarly $\mathit{fw}(abcadcbd) =4$.

\begin{definition} An $(r, s)$-formation $f$ is called {\it binary} if there exists a permutation $p$ on $r$ letters such that every permutation in $f$ is either the same as $p$ or the reverse of $p$. \end{definition}

Most of the proofs in this paper depend on the fact that if $u$ is a sequence with $r$ distinct letters, then every binary $(r, s)$-formation contains $u$ if and only if $s \geq \mathit{fw}(u)$. We use the following notation to describe binary formations more concisely.

\begin{definition} $I_{c}$ is the increasing sequence $1 \ldots c$ on $c$ letters and $D_{c}$ is the decreasing sequence $c\ldots 1$ on $c$ letters. Given a permutation $\pi \in S_{c}$, the sequences $I_{\pi}$ and $D_{\pi}$ are $\pi(1)\ldots \pi(c)$ and $\pi(c)\ldots \pi(1)$ respectively.\end{definition} 

We focus especially on two classes of binary formations in order to derive bounds on $\mathit{fw}(u)$. The sequence $\mathit{up}(l, t)$ is $I_{l}$ repeated $t$ times, and $\mathit{alt}(l, t)$ is a concatenation of $t$ permutations, starting with $I_{l}$ and alternating between $I_{l}$ and $D_{l}$. For example, $\mathit{up}(3,3)=123123123$ and $\mathit{alt}(3,3)=123321123$.

\begin{definition} If $u$ is a sequence with $c$ distinct letters, then $\mathit{l}(u)$ is the smallest $k$ such that $\mathit{up}(c, k)$ contains $u$, and $\mathit{r}(u)$ is the smallest $k$ such that $\mathit{alt}(c, k)$ contains $u$.\end{definition}

Then $\mathit{fw}(u)\geq \mathit{l}(u)$ and $\mathit{fw}(u)\geq \mathit{r}(u)$. We evaluate both $\mathit{l}(u)$ and $\mathit{r}(u)$ for every binary formation $u$.

In Section \ref{sec2.1} we prove that $\gamma(r, s) =(r-1)^{2^{s-1}}+1$ is the minimum value for which every $(\gamma(r, s), s)$-formation contains a binary $(r, s)$-formation. It follows that if $u$ has $r$ distinct letters, then $\mathit{fw}(u)$ is the minimum $s$ for which every binary $(r, s)$-formation contains $u$.

In Section \ref{sec2.2} we prove that $\mathit{fw}(u)=t-1$ for every sequence $u$ with two distinct letters and length $t$. We also determine every sequence $u$ for which $\mathit{fw}(u)\leq 3$. In addition, we show that $\mathit{fw}(\mathit{up}(c, t))=2t-1$ for all $c\geq 2$ and $t\geq 1$. This implies that $\mathit{Ex}(\mathit{up}(l, t), n)=n2^{\frac{1}{(t-2)!}\alpha(n)^{t-2}\pm O(\alpha(n)^{t-3})}$ for all $l \geq 2$ and $t\geq 3$ and that $\mathit{fw}(u)\leq 2\mathit{l}(u)-1$ for every sequence $u$.

In Section \ref{sec3.1} we compute $\mathit{l}(u)$ and use the result to bound $\mathit{fw}(u)$ up to a factor of $2$ for every binary formation $u$. In particular we prove the following bounds on $\mathit{fw}(u)$.

\begin{theorem} \label{1.5} 
Fix $c\geq 2$ and let $u=I_{c}^{e_{1}}D_{c}^{e_{2}}I_{c}^{e_{3}}\ldots \mathcal{L}_{c}^{e_{n}}$, where $\mathcal{L}$ is $I$ if $n$ is odd and $D$ if $n$ is even, and $e_{i}>0$ for all $i$. Define $A= \sum_{i\geq 1}e_{2i-1}$ and $B = \sum_{i\geq 1}e_{2i}$. Let $M= \max(A, B)$ and let $m= \min(A, B)$. Then $(c-1)m+M+ \lfloor \frac{n}{2}\rfloor\leq \mathit{fw}(u)\leq 2(c-1)m+2M+2\lfloor\frac{n}{2}\rfloor-1$.
\end{theorem}

In Section \ref{sec3.2} we compute $\mathit{r}(u)$ for every binary formation $u$. Specifically we prove that if $c\geq 2$, then $\mathit{r}(I_{c}^{e_{1}}D_{c}^{e_{2}}I_{c}^{e_{3}} \ldots \mathcal{L}_{c}^{e_{n}})=2\sum_{i=1}^{n}e_{i}-n$, where $\mathcal{L}$ is $I$ if $n$ is odd and $D$ if $n$ is even. 

In Section \ref{sec4} we use $\mathit{fw}(u)$ to derive tight bounds on $\mathit{Ex}(u, n)$ for other sequences $u$ besides $\mathit{up}(l, t)$. Let $u$ be any sequence of the form $a v a v' a$ such that $a$ is a letter, $v$ is a nonempty sequence excluding $a$ with no repeated letters and $v'$ is obtained from $v$ by only moving the first letter of $v$ to another place in $v$. We show that $\mathit{fw}(u)=4$, implying that $\mathit{Ex}(u, n)=\Theta(n\alpha(n))$. We also prove that $\mathit{Ex}(abc(acb)^{t}, n)=n2^{\frac{1}{(t-1)!}\alpha(n)^{t-1}\pm O(\alpha(n)^{t-2})}$ for all $t\geq 2$.

In Section \ref{sec5} we compute $\mathit{fw}$ for various classes of binary formations. In particular we show for $c\geq 2$ and $k\geq 1$ that $\mathit{fw}(I_{c}D_{c}I_{c}) =c+3$, $\mathit{fw}(I_{c}^{k}D_{c})= c+2k-1$, $\mathit{fw}(I_{c}D_{c}I_{c}D_{c}) = 2c+3$, $\mathit{fw}(\mathit{alt}(c, 2k))\geq k(c+2)-1$, and $\mathit{fw}(\mathit{alt}(c, 2k+1))\geq k(c+2)+1$. 

In Section \ref{sec6} we discuss some unresolved questions.

\section{An extension of the Erdos-Szekeres theorem} \label{sec2.1}

The following upper bound is obtained by iterating the Erdos-Szekeres theorem as in \cite{6}.

\begin{lemma} \label{2.1} Every $((r-1)^{2^{s-1}}+1, s)$-formation contains a binary $(r, s)$-formation. \end{lemma}

\begin{proof} We prove by induction on $s$ that every $((r-1)^{2^{s-1}}+1, s)$-formation contains a binary $(r, s)$-formation. Clearly this is true for $s =1$. For the inductive hypothesis fix $s$ and suppose for every $r\geq 1$ that each $((r-1)^{2^{s-1}}+1, s)$-formation contains a binary $(r, s)$-formation. 

Consider any $((r-1)^{2^{s}}+1, s +1)$-formation $F$. Without loss of generality suppose that the first permutation of $F$ is $I_{(r-1)^{2^{s}}+1}$. By inductive hypothesis the first $s$ permutations of $F$ contain a binary $((r-1)^{2}+1, s)$-formation $f$. By the Erdos-Szekeres theorem, every sequence of $(x-1)^{2}+1$ distinct integers contains an increasing or decreasing subsequence of length $x$. Therefore the last permutation of $F$ contains an increasing or decreasing subsequence of length $r$ on the letters of $f$. Thus $F$ contains a binary $(r, s+1)$-formation. 
\end{proof}

\begin{corollary} \label{2.2} If $u$ has $r$ distinct letters, then every binary $(r, s)$-formation contains $u$ if and only if $s \geq \mathit{fw}(u)$. \end{corollary}

\begin{proof} If for some $s$ every binary $(r, s)$-formation contains $u$, then there exists a function $\gamma(r, s)$ such that every $(\gamma(r, s), s)$-formation contains $u$. Thus $\mathit{fw}(u)\leq s$.

If some binary $(r, s -1)$-formation $f$ avoids $u$, then for every $z\geq r$ the binary $(z, s -1)$-formations which contain $f$ will avoid $u$. Hence $\mathit{fw}(u)> s -1$.
\end{proof}

\begin{corollary} \label{2.3} If $u$ is a nonempty sequence and $v$ is obtained from $u$ by inserting a single occurrence of a letter which has no occurrence in $u$, then $\mathit{fw}(u)= \mathit{fw}(v)$.\end{corollary}

\begin{proof} If $u$ has $r$ distinct letters, then every binary $(2r+1, \mathit{fw}(u))$-formation $F$ with first permutation $I_{2r+1}$ has a copy of $u$ using only the even numbers $2, \ldots, 2r$. Since there is at least one odd number between every pair of even numbers in $F$, then the copy of $u$ in $F$ can be extended to a copy of $v$ using an odd number. \end{proof}

\begin{corollary} \label{2.4} If $u$ has $r$ distinct letters, then $\mathit{fl}(u)\leq(r-1)^{2^{\mathit{fw}(u)-1}}+1$.\end{corollary}

\begin{proof} Since every binary $(r, \mathit{fw}(u))$-formation contains $u$, then every $((r- 1)^{2^{\mathit{fw}(u)-1}}+1, \mathit{fw}(u))$-formation contains $u$.\end{proof}

The next theorem shows that the upper bound in Lemma \ref{2.1} is tight.

\begin{theorem} \label{2.5} For every $r, s \geq 1$ there exists a $((r-1)^{2^{s-1}}, s)$-formation that avoids every binary $(r, s)$-formation.\end{theorem}

\begin{proof} We construct the desired formation $\mathit{F}(r, s)$ one permutation at a time. Define an $\alpha$-{\it block} in $\mathit{F}(r, s)$ to be a block of numbers in a permutation from positions $(k-1)(r-1)^{\alpha}+1$ to $k(r-1)^{\alpha}$ for some $k$. For $k\leq s -1$ define a $k$-{\it swap} on a permutation of length $(r-1)^{2^{s-1}}$ as follows: For every even $i, 1<i\leq 2^{k}$, a $k$-swap reverses the placement of the $(i-1)2^{s-k-1}$-blocks in each $i2^{s-k-1}$-block. For example if $(r, s) =(3,3)$, then a $1$-swap on $1234567890ABCDEF$ produces $CDEF90AB56781234$.

Let permutation 1 of $\mathit{F}(r, s)$ be the identity permutation on the letters $1, \ldots, (r-1)^{2^{s-1}}$. To form permutation $k+1$ of $\mathit{F}(r, s)$, perform a $k$-swap on permutation $k$. The next lemma about $\mathit{F}(r, s)$ will imply that $\mathit{F}(r, s)$ avoids every binary $(r, s)$-formation.

\begin{lemma} \label{2.6} Consider any set $B$ of distinct numbers occurring in each of the first $k$ permutations of $\mathit{F}(r, s)$ with the same or reverse order in adjacent permutations. Let $i(k)= \sum_{j=1}^{k-1}e_{j}2^{k-j-1}$ where $e_{j}=1$ if the elements in $B$ reverse order from permutation $j$ to permutation $j+1$ and $e_{j}=0$ otherwise. Then in permutation $k$ the elements of $B$ are contained in different $i(k)2^{s-k}$-blocks, but the same $(i(k)+1)2^{s-k}$-block. \end{lemma}

\begin{proof} We induct on $k$. When $k=1$, $i(k)=0$. The entire permutation is a $2^{s-1}$-block and $0$-blocks are individual elements, so the lemma is true when $k = 1$.

For the inductive hypothesis, suppose that in permutation $k$ the elements of $B$ are contained in a single $(i(k)+1)2^{s-k}$-block but different $i(k)2^{s-k}$-blocks. Consider any set $B$ of distinct numbers occurring in each of the first $k+1$ permutations of $\mathit{F}(r, s)$ with the same or reverse order in adjacent permutations.

Now consider the $k$-swap that sends permutation $k$ of $\mathit{F}(r, s)$ to permutation $k+1$. The parts of the swap that reverse the placement of the $(j-1)2^{s-k-1}$-blocks in each $j2^{s-k-1}$-block for even $j\geq 2i(k)+4$ do not affect the order of the elements of $B$ since the elements of $B$ are contained in a single $(2i(k)+2)2^{s-k-1}$-block.

The parts of the swap that reverse the placement of the $(j-1)2^{s-k-1}$-blocks in each $j2^{s-k-1}$-block for even $j\leq 2i(k)$ also do not affect the order of the elements of $B$ since the elements of $B$ are contained in different $(2i(k))2^{s-k-1}$-blocks. Thus the only part of the swap which is relevant to the order of the elements in $B$ is the reversal of the placement of the $(2i(k)+1)2^{s-k-1}$-blocks inside each $(2i(k)+2)2^{s-k-1}$-block.

If the order of elements in $B$ reverses from permutation $k$ to permutation $k+1$, then $i(k+1)=2i(k)+1$. All the elements of $B$ must be contained in different $(2i(k)+1)2^{s-k-1}$-blocks, or else the $k$-swap would not reverse their order. By inductive hypothesis the elements of $B$ are contained in the same $(i(k+1)+1)2^{s-k-1}$-block.

If the order of elements in $B$ is the same in permutation $k$ and permutation $k+1$, then $i(k+1)=2i(k)$. The elements of $B$ must be contained in the same $(2i(k)+1)2^{s-k-1}$-block, or else the $k$-swap would not preserve their order. By inductive hypothesis the elements of $B$ are contained in different $i(k+1)2^{s-k-1}$-blocks.
\end{proof}

Given any set $B$ of distinct numbers contained in every permutation of $\mathit{F}(r, s)$ whose order either stays the same or reverses between adjacent permutations, there is some $i$ such that the elements of $B$ are in different $i$-blocks, but the same $(i+1)$-block of permutation $s$. Since there are $r-1$ $i$-blocks in each $(i+1)$-block, then $r-1$ is the maximum possible number of elements in $B$.
\end{proof}

\section{Using binary formations to compute $\mathit{fw}$}\label{sec2.2}

If $u$ has one distinct letter, then $\mathit{fw}(u)$ is the length of $u$. If $u$ has two distinct letters, then $\mathit{fw}(u)$ also depends only on the length of $u$.

\begin{lemma} \label{2.7} If $u$ has two distinct letters and length $t$, then $\mathit{fw}(u)=t-1$.
\end{lemma}

\begin{proof} By Lemma \ref{1.3} it suffices to prove this lemma for sequences with different first and second letters. The upper bound follows since every $(2, t- 1)$-formation contains $u$. For the lower bound it suffices to construct a $(2, t-1)$ formation $f(u)$ which only contains copies of $u$ for which the last letter of the copy of $u$ is the last letter of $f(u)$. Therefore the $(2, t-2)$-formation in the first $t-2$ permutations of $f(u)$ avoids $u$, so $\mathit{fw}(u)>t-2$ by Corollary \ref{2.2}.

Assume without loss of generality that $u$ starts with $xy$. Construct $f(u)$ by ignoring the leading $x$ and replacing every $x$ in $u$ by $b a$ and every $y$ by $a b$. Let $u'$ denote the sequence obtained by deleting the last letter of $u$. We prove by induction on the length of $u$ that $f(u)$ contains only copies of $u$ for which the last letter of the copy of $u$ is the last letter of $f(u)$. The first case to consider is $u=xy$.

Since $f(xy)= a b$, then $f(xy)$ contains exactly one copy of the sequence $xy$ and the last letter of the copy of $xy$ is the last letter of $f(xy)$. Suppose by inductive hypothesis that $f(u')$ contains only copies of $u'$ for which the last letter of the copy of $u'$ is the last letter of $f(u')$. If the last two letters of $u$ are the same, then the first letter of the last permutation of $f(u)$ is different from the last letter of $f(u')$, so the last letter of $f(u)$ will be the last letter of any copy of $u$ in $f(u)$. If the last two letters of $u$ are different, then the first letter of the last permutation of $f(u)$ is the same as the last letter of $f(u')$, so the last letter of $f(u)$ will be the last letter of any copy of $u$ in $f(u)$.
\end{proof}

If $u$ has at least three distinct letters, then $\mathit{fw}(u)$ cannot be determined solely from the length of $u$ and the number of distinct letters in $u$. For example $\mathit{fw}(abcabc)=3$ and $\mathit{fw}(abccba)=4$.

The next lemma identifies all sequences $u$ for which $\mathit{fw}(u)=3$. As a result of Corollary \ref{2.3}, deleting any letters which occur just once in $u$ will not change the value of $\mathit{fw}(u)$ unless $u$ has no other letters. We call a sequence {\it reduced} if every distinct letter in the sequence occurs at least twice.

By Lemma \ref{2.7}, $\mathit{fw}(u)=1$ if and only if $u$ is nonempty and each distinct letter in $u$ occurs once, and $\mathit{fw}(u)=2$ if and only if one letter in $u$ occurs twice and every other distinct letter occurs once.

\begin{lemma} \label{2.8} If $u$ is reduced and $\mathit{fw}(u)=3$, then either there exists some $l \geq 2$ for which $u$ is isomorphic to $\mathit{up}(l, 2)$ or $u$ is isomorphic to one of the sequences $aaa$, $aabb$, $abba$, $abcacb$, $abcbac$, $abccab$, or $abcdbadc$.
\end{lemma}

\begin{proof} Since $u$ is reduced, then every distinct letter in $u$ occurs at least twice. If any letter in $u$ occurs three times, then it is the only letter in $u$ and $u$ is isomorphic to $aaa$, or else $\mathit{fw}(u)\geq 4$ by Lemma \ref{2.7}. If $u$ is not isomorphic to $aaa$, then every distinct letter in $u$ occurs twice.

Suppose $u$ is not isomorphic to $\mathit{up}(l, 2)$ for any $l \geq 2$. Then there exist two distinct letters $x$ and $y$ in $u$ for which the subsequence consisting of occurrences of $x$ and $y$ is isomorphic to $aabb$ or $abba$. If $x$ and $y$ are the only distinct letters in $u$, then $u$ is isomorphic to $aabb$ or $abba$.

If $u$ has three distinct letters, then $u$ is isomorphic to a sequence obtained by adding two occurrences of $c$ anywhere in $aabb$ or $abba$, so we consider $30$ cases. If $u$ had the form $xxv$ or $vxx$ for some letter $x$ and sequence $v$ of length $4$ with two distinct letters not equal to $x$, then $\mathit{fw}(v)=3$ by Lemma \ref{2.7}, so $\mathit{fw}(u)=4$ by Lemma \ref{1.3}. This eliminates the cases $aabbcc$, $aabcbc$, $aacbbc$, $aabccb$, $aacbcb$, $aaccbb$, $acacbb$, $caacbb$, $accabb$, $cacabb$, $ccaabb$, $abbacc$, and $ccabba$.

The binary $(3,3)$-formation $xyzxyzxyz$ avoids $caabbc$, $abbcca$, $accbba$, $cabbac$, $acbbca$, and $abccba$. The binary $(3,3)$-formation $xyzzyxxyz$ avoids $acabcb$, $abcbca$, and $acbcba$. The binary $(3, 3)$-formation $xyzzyxzyx$ avoids $acabbc$ and $cacbba$. So its reverse avoids $caabcb$ and $abbcac$. Thus each of these sequences have formation width at least $4$ by Corollary \ref{2.2}.

If $u$ is one of the remaining sequences $abcbac$, $acbbac$, $cabbca$, or $cabcba$, then $\mathit{fw}(u)=3$. Thus every reduced sequence $u$ with three distinct letters for which $\mathit{fw}(u)=3$ is a $(3, 2)$-formation. Note that $acbbac$ and $cabbca$ are isomorphic to $abccab$, and $cabcba$ is isomorphic to $abcacb$.

If $u$ has four distinct letters, then $u$ is isomorphic to a sequence obtained by adding two occurrences of $d$ to the sequence $abcabc$, $abcacb$, $abcbac$, or $abccab$. If $u$ was not a $(4, 2)$-formation, then $u$ would contain a reduced sequence $v$ with three distinct letters which was not a $(3, 2)$-formation, so $\mathit{fw}(u)\geq 4$.

We consider each $(4, 2)$-formation with first permutation $abcd$. The binary $(4, 3)$-formation $xyzwxyzwxyzw$ avoids $abcdadcb$, $abcdbdca$, $abcdcbad$, $abcdcbda$, $abcddacb$, $abcd\-dbac$, $abcddbca$, $abcddcab$, and $abcddcba$. The binary $(4,3)$-formation $xyzwxyzwwzyx$ avoids $abcdbacd$, $abcdcabd$, and $abcdcadb$. The binary $(4,3)$-formation $xyzwwzyxxyzw$ avoids $abcd\-acbd$, $abcd\-acdb$, $abcd\-bcad$, $abcd\-bcda$, $abcd\-cdab$, and $abcd\-cdba$. The binary $(4, 3)$-formation $xyzw\-wzyx\-wzyx$ avoids $abcd\-abdc$, $abcd\-adbc$, $abcd\-bdac$, and $abcd\-dabc$. Thus each of these $(4, 2)$-formations have formation width at least $4$ by Corollary \ref{2.2}.

If $u$ is $abcdbadc$, then $\mathit{fw}(u)=3$. If $u$ had five distinct letters, then $u$ must be a $(5, 2)$-formation or else $\mathit{fw}(u)\geq 4$. If $u$ was any $(5, 2)$-formation with first permutation $abcde$, then every $(4, 2)$-formation in $u$ would be isomorphic to $abcdbadc$ or $\mathit{up}(4, 2)$. It is impossible for a $(5, 2)$-formation to have both a subsequence isomorphic to $abcdbadc$ and another subsequence isomorphic to $\mathit{up}(4, 2)$, so every $(4, 2)$-formation in $u$ would be isomorphic to $abcdbadc$ or else $u$ would be isomorphic to $\mathit{up}(5, 2)$. In particular $u$ must have both $abcdbadc$ and $acdecaed$ as subsequences, a contradiction.

If $u$ had $r$ distinct letters for some $r > 5$ and $u$ was not isomorphic to $\mathit{up}(r, 2)$, then $u$ would contain a subsequence of length $10$ with five distinct letters that was not isomorphic to $\mathit{up}(5, 2)$, so $\mathit{fw}(u) > 3$. \end{proof}

The last lemma can also be verified by using the formation width algorithm in the appendix. The next lemma provides an upper bound on $\mathit{fw}(u)$ for every binary formation $u$. It is tight if $u=\mathit{up}(l, t)$ for any $l \geq 2$ and $t\geq 1$.

\begin{lemma} \label{2.9} Let $u=I_{c}^{e_{1}}D_{c}^{e_{2}}I_{c}^{e_{3}}\ldots \mathcal{L}_{c}^{e_{n}}$, where $\mathcal{L}$ is $I$ if $n$ is odd and $D$ if $n$ is even so that $e_{i}>0$ for each $i$ and $ \sum_{i=1}^{n}e_{i}=k$. Then $\mathit{fw}(u)\leq c(k-e_{m})+2e_{m}-1$ for all $m$.
\end{lemma}

\begin{proof} Let $k_{1}= \sum_{i=1}^{m-1}e_{i}$ and $k_{2}= \sum_{i=m+1}^{j}e_{i}$. In any binary $(c, c(k-e_{m})+ 2e_{m}-1)$-formation $f$, there is a copy of $\mathit{up}(c, e_{m})$ in permutations $ck_{1}+1$ through $ck_{1}+2e_{m}-1$ of $f$ by the pigeonhole principle. This copy of $\mathit{up}(c, e_{m})$ can be extended to make a copy of $u$ in $f$ by using one letter from each of the remaining $ck_{1}+ck_{2}$ permutations of $f$. Thus $\mathit{fw}(u)\leq c(k-e_{m})+2e_{m}-1$ by Corollary \ref{2.2}. 
\end{proof}

\begin{theorem} \label{2.10} $\mathit{fw}(\mathit{up}(l, t))=2t-1$ for every $l \geq 2$ and $t\geq 1$.
\end{theorem}

\begin{proof} For the lower bound $\mathit{fw}(\mathit{up}(l, t))\geq \mathit{fw}((a b)^{t})=2t-1$ since $\mathit{up}(l, t)$ contains $(a b)^{t}$. The upper bound $\mathit{fw}(\mathit{up}(l, t))\leq 2t-1$ follows from Lemma \ref{2.9}. 
\end{proof}

Therefore $\mathit{fw}(u)=2t-1$ for every sequence $u$ such that $u$ contains $(a b)^{t}$ and there exists $l \geq 2$ for which $\mathit{up}(l, t)$ contains $u$. As a corollary this implies the upper bounds in the next result, which gives nearly tight asymptotic bounds on $\mathit{Ex}(\mathit{up}(l, t), n)$. The lower bounds in the next corollary follow from the lower bounds on $\mathit{Ex}((a b)^{t}, n)$ in \cite{2} by Lemma \ref{1.1}.

\begin{corollary} \label{2.11} $\mathit{Ex}(\mathit{up}(l, t), n)=n2^{\frac{1}{(t-2)!}\alpha(n)^{t-2}\pm O(\alpha(n)^{t-3})}$ for all $l \geq 2$ and $t\geq 3$.
\end{corollary}

As a result, the constant $c$ improves in the $(n\log n)2^{\alpha(n)^{c}}$ upper bound from \cite{10} on the maximum number of edges in $k$-quasiplanar graphs on $n$ vertices with no pair of edges intersecting in more than $O(1)$ points, since their proof used the bounds $\mathit{Ex}(\mathit{up}(l, t), n)\leq nl2^{lt-3}(10l)^{10\alpha(n)^{lt}}$ from \cite{6}.

\section{Bounding the formation width of binary formations}\label{sec3}

In this section we compute the exact values of $\mathit{l}(u)$ and $\mathit{r}(u)$ for all binary formations $u$. This yields upper and lower bounds on $\mathit{fw}(u)$ which differ by at most a factor of two for each binary formation $u$.

\subsection{Computing $l$}\label{sec3.1}

 If $\pi\in S_{c}$ and $u$ is a sequence on the letters $1, \ldots, c$, then let $\mathit{l}_{\pi}(u)=k$ if $u$ is a subsequence of $I_{\pi}^{k}$ but $u$ is not a subsequence of $I_{\pi}^{k-1}$. It follows that $\mathit{l}(u)= \min_{\pi\in S_{c}}\{\mathit{l}_{\pi}(u)\}$.

\begin{lemma} \label{3.1} If $\mathit{l}_{\pi}(I_{c})= a$ and $\mathit{l}_{\pi}(D_{c})=b$, then $a+b=c+1$.
\end{lemma}

\begin{proof} Represent the permutation $\pi$ by the set of points $(i, \pi(i))$. Connect points $(i, \pi(i))$ and $(j, \pi(j))$ if $i<j$ and $\pi(j)=\pi(i)+1$. This partitions the points into $a$ connected sections. In a different representation connect points $(i, \pi(i))$ and $(j, \pi(j))$ if $i<j$ and $\pi(j)=\pi(i)-1$. This partitions the points into $b$ connected sections.

We count the total number of endpoints of connected sections of points in both representations in two ways so that each connected section of points is considered to have two endpoints, even when the section consists of a single point. Since every connected section has two endpoints, then there are a total of $2(a+b)$ endpoints. Alternatively every point $(i, \pi(i))$ contributes two endpoints, unless $\pi(i)=1$ or $\pi(i)=c$, in which case $(i, \pi(i))$ contributes three endpoints. Thus there are a total of $2c+2$ endpoints, so $a+b= c+1$.
\end{proof}

\begin{corollary} \label{3.2} $\mathit{l}(I_{c}D_{c})=c+1$ for every $c\geq 1$.\end{corollary}

\begin{corollary} \label{3.3} $\mathit{fw}(I_{c}D_{c})=c+1$ for every $c\geq 1$.\end{corollary}

\begin{proof} The upper bound is trivial. The lower bound follows since $I_{c}^{c}$ avoids $I_{c}D_{c}$.
\end{proof}

If $u$ and $v$ are sequences on the letters $1, \ldots, c$, then $\mathit{l}_{\pi}(u)+\mathit{l}_{\pi}(v)-1\leq \mathit{l}_{\pi}(uv)\leq \mathit{l}_{\pi}(u)+\mathit{l}_{\pi}(v)$. Say that $u$ and $v$ $\pi$-{\it overlap} if $\mathit{l}_{\pi}(uv)=\mathit{l}_{\pi}(u)+\mathit{l}_{\pi}(v)-1$. Then $u$ and $v$ $\pi$-overlap if and only if the last letter of $u$ and the first letter of $v$ $\pi$-overlap.

For each $\pi\in S_{c}$, the sequences $I_{c}$ and $D_{c}$ do not $\pi$-overlap since the last letter of $I_{c}$ is the first letter of $D_{c}$, and $D_{c}$ and $I_{c}$ do not $\pi$-overlap since the last letter of $D_{c}$ is the first letter of $I_{c}$. Furthermore if $c\geq 2$, then exactly one of the two sequences $I_{c}$ or $D_{c}$ $\pi$-overlaps itself, depending on the order in which the first and last letters of $I_{c}$ occur in $I_{\pi}$. Moreover for any sequence $u$, if $\mathit{l}_{\pi}(u)=1$ then $u$ does not $\pi$-overlap itself.

The next theorem implies Theorem \ref{1.5} since $\mathit{l}(u)\leq \mathit{fw}(u)\leq 2\mathit{l}(u)-1$.

\begin{theorem}\label{3.4} Fix $c\geq 2$ and let $u=I_{c}^{e_{1}}D_{c}^{e_{2}}I_{c}^{e_{3}}\ldots \mathcal{L}_{c}^{e_{n}}$, where $\mathcal{L}$ is $I$ if $n$ is odd and $D$ if $n$ is even, and $e_{i}>0$ for all $i$. Define $A= \sum_{i\geq 1}e_{2i-1}$ and $B = \sum_{i\geq 1}e_{2i}$. Let $M= \max(A, B)$ and let $m= \min(A, B)$. Then $ \mathit{l}(u)=(c-1)m+M+ \lfloor\frac{n}{2}\rfloor$.\end{theorem}

\begin{proof} Fix an arbitrary $\pi\in S_{c}$ and let $ \mathit{l}_{\pi}(I_{c})=a$ and $\mathit{l}_{\pi}(D_{c})=b$. We show $\mathit{l}_{\pi}(u) \geq (c-1)m+M+\lfloor\frac{n}{2}\rfloor$ by considering two cases depending on whether $I_{c}$ or $D_{c}$ $\pi$-overlaps itself.

\textbf{Case 1:} $I_{c}$ $\pi$-overlaps itself.

In this case $\mathit{l}_{\pi}(I_{c}^{e_{i}})=(a-1)e_{i}+1$ and $\mathit{l}_{\pi}(D_{c}^{e_{i}})=b e_{i}$. Since $I_{c}$ and $D_{c}$ do not $\pi$-overlap and $D_{c}$ and $I_{c}$ do not $\pi$-overlap, then $\mathit{l}_{\pi}(u)=(a-1)A+bB+ \lceil\frac{n}{2}\rceil$. Lemma \ref{3.1} implies that $(a-1)+b=c$, while $b>0$ and $a >1$ since $I_{c}$ $\pi$-overlaps itself. Then $\mathit{l}_{\pi}(u) \geq (c-1)m+M+ \lceil\frac{n}{2}\rceil$.

\textbf{Case 2:} $D_{c}$ $\pi$-overlaps itself.

In this case $\mathit{l}_{\pi}(I_{c}^{e_{i}})=a e_{i}$ and $\mathit{l}_{\pi}(D_{c}^{e_{i}})=(b-1)e_{i}+1$, so $\mathit{l}_{\pi}(u)=a A + (b-1)B+\lfloor \frac{n}{2}\rfloor$. Moreover $a +(b-1)=c$, $a > 0$, and $b>1$. Then $\mathit{l}_{\pi}(u) \geq (c-1)m+M+ \lfloor \frac{n}{2}\rfloor$.

Thus in either case $\mathit{l}_{\pi}(u) \geq(c-1)m+M+\lfloor \frac{n}{2}\rfloor$. If $A\geq B$, then this value is attained by letting $\pi$ be the identity permutation. If $B>A$, then this value is attained by letting $\pi(1)=1$ and $\pi(i)=c+2-i$ for $2\leq i\leq c$.
\end{proof}

\subsection{Computing $r$}\label{sec3.2}

For every binary formation $u$ we compute $\mathit{r}(u)$, and we identify when $\mathit{r}(u)> \mathit{l}(u)$.

\begin{theorem} \label{3.5} If $c\geq 2$ and $e_{i}>0$ for all $i$, then $\mathit{r}(I_{c}^{e_{1}}D_{c}^{e_{2}}I_{c}^{e_{3}}\ldots \mathcal{L}_{c}^{e_{n}})= 2 \sum_{i=1}^{n}e_{i}-n$, where $\mathcal{L}$ is $I$ if $n$ is odd and $D$ if $n$ is even.
\end{theorem}

\begin{proof} First we show that $\mathit{r}(I_{c}^{x})=2x-1$ for every $x>0$. The upper bound is trivial. For the lower bound we also show that $\mathit{alt}(c, 2x-1)$ has the subsequence $I_{\pi}^{x}$ only if $\pi(c)=c$.

We proceed by induction on $x$. Clearly $\mathit{r}(I_{c})=1$. In addition, $I_{\pi}$ is a subsequence of $I_{c}$ only if $\pi$ is the identity permutation, so $\pi(c)=c$. For the inductive hypothesis assume that $\mathit{r}(I_{c}^{x})=2x-1$ and that $\mathit{alt}(c, 2x-1)$ has the subsequence $I_{\pi}^{x}$ only if $\pi(c)=c$. We claim that $\mathit{r}(I_{c}^{x+1})=2x+1$ and that $\mathit{alt}(c, 2x+1)$ has the subsequence $I_{\pi}^{x+1}$ only if $\pi(c)=c$.

Let $\pi$ be an arbitrary permutation. We will first show that $I_{\pi}^{x+1}$ is not a subsequence of $\mathit{alt}(c, 2x)$. Suppose for contradiction that $I_{\pi}^{x+1}$ is a subsequence of $\mathit{alt}(c, 2x)$. Then $I_{\pi}^{x}$ is a subsequence of $\mathit{alt}(c, 2x-1)$, so $\pi(c)=c$. Then the last letter in $I_{\pi}^{x+1}$ must be the first letter of the last permutation of $\mathit{alt}(c, 2x)$, a contradiction. Thus $\mathit{r}(I_{c}^{x+1})=2x+1$. We still must show that $\mathit{alt}(c, 2x+1)$ has the subsequence $I_{\pi}^{x+1}$ only if $\pi(c)=c$.

Suppose $\pi(c)=i$ for some $1\leq i<c$, and assume for contradiction that $I_{\pi}^{x+1}$ is a subsequence of $\mathit{alt}(c, 2x+1)$. Since $I_{\pi}^{x}$ is not a subsequence of $\mathit{alt}(c, 2x-1)$, then the second to last $i$ in $I_{\pi}^{x+1}$ must occur in the second to last permutation of $\mathit{alt}(c, 2x+1)$ and the last $i$ in $I_{\pi}^{x+1}$ must occur in the last permutation of $\mathit{alt}(c, 2x+1)$. Since $i<c$, then there are at most $c-2$ distinct letters between the occurrences of $i$ in the last two permutations of $\mathit{alt}(c, 2x+1)$, a contradiction. Thus $\mathit{alt}(c, 2x+1)$ has the subsequence $I_{\pi}^{x+1}$ only if $\pi(c)=c$. This completes the induction.

By symmetry we find that $\mathit{r}(D_{c}^{x})=2x-1$ for every $x>0$. We now prove the claim that $\mathit{r}(I_{c}^{e_{1}}D_{c}^{e_{2}}I_{c}^{e_{3}}\ldots \mathcal{L}_{c}^{e_{n}})= 2 \sum_{i=1}^{n}e_{i}-n$. The upper bound is trivial since the copy of $I_{c}^{e_{1}}D_{c}^{e_{2}}I_{c}^{e_{3}}\ldots \mathcal{L}_{c}^{e_{n}}$ can be selected greedily from left to right in $\mathit{alt}(c, 2\sum_{i=1}^{n}e_{i}-n)$. For the lower bound, suppose for some $k$ and permutation $\pi$ that $\mathit{alt}(c, k)$ has the subsequence $I_{\pi}^{e_{1}}D_{\pi}^{e_{2}}I_{\pi}^{e_{3}}\ldots \mathcal{L}_{\pi}^{e_{n}}$ with $n$ sections of the form $I_{\pi}^{x}$ or $D_{\pi}^{x}$. No section $I_{\pi}^{x}$ or $D_{\pi}^{x}$ can occur in fewer than $2x-1$ adjacent permutations of $\mathit{alt}(c, k)$. Furthermore no different sections have letters occurring in the same permutation. Thus $\mathit{alt}(c, k)$ contains at least $2  \sum_{i=1}^{n}e_{i}-n$ permutations, so $k \geq 2\sum_{i=1}^{n}e_{i}-n$.
\end{proof}

\begin{corollary} \label{3.6} Fix $c\geq 2$ and let $u=I_{c}^{e_{1}}D_{c}^{e_{2}}I_{c}^{e_{3}}\ldots \mathcal{L}_{c}^{e_{n}}$, where $\mathcal{L}$ is $I$ if $n$ is odd and $D$ if $n$ is even, and $e_{i}>0$ for all $i$. Define $A= \sum_{i\geq 1}e_{2i-1}$ and $B = \sum_{i\geq 1}e_{2i}$. Let $M= \max(A, B)$ and let $m= \min(A, B)$. Then $\mathit{r}(u)>\mathit{l}(u)$ if and only if $ M>(c-3)m+n+\lfloor\frac{n}{2}\rfloor$.
\end{corollary}

\begin{proof} This follows from setting $2  \sum_{i=1}^{n}e_{i}-n>(c-1)m+M+\lfloor\frac{n}{2}\rfloor$ since $  \sum_{i=1}^{n}e_{i}=m+M$.
\end{proof}

\section{Further bounds on extremal functions using $\mathit{fw}$}\label{sec4}

The lemmas in this section use Corollary \ref{2.2} to identify sequences $u$ with $\mathit{fw}(u)>3$ for which $\mathit{fw}(u)$ provides tight upper bounds on $\mathit{Ex}(u, n)$, starting with an infinite set of sequences which contain $a b a b a$.

\begin{lemma}\label{4.1} If $u$ is any sequence of the form $a v a v' a$ such that $a$ is a letter, $v$ is a nonempty sequence excluding $a$ with no repeated letters and $v'$ is obtained from $v$ by only moving the first letter of $v$ to another place in $v$, then $\mathit{fw}(u)=4$.
\end{lemma}

\begin{proof} Since $u$ contains an alternation of length $5$, then $\mathit{fw}(u)\geq 4$. Suppose $u$ has $r$ distinct letters for $r\geq 2$. In order to prove that $\mathit{fw}(u)\leq 4$, it suffices by Corollary \ref{2.2} to show that $u$ is contained in every binary $(r, 4)$-formation. First note that binary $(r, 4)$-formations isomorphic to $I_{r}^{4}$ or $I_{r}^{3}D_{r}$ contain a copy of $u$ which uses every letter in the first permutation.

Furthermore if the position in $v'$ of the occurrence of the first letter of $v$ is right after the occurrence in $v'$ of the $i^{th}$ letter of $v$, then $I_{r}^{2} D_{r} I_{r}$ has a subsequence $u'$ isomorphic to $u$ such that the $j^{th}$ letter of $u'$ is given by $(r-i+j-2 \mod{r}) + 1$ for each $1\leq j\leq r$. In particular the subsequence $u'$ includes the last $i+1$ letters in the first permutation of $I_{r}^{2} D_{r} I_{r}$, all of the letters except $r-i+1$ in the second permutation, the single letter $r-i+1$ in the third permutation, and the first $r-i$ letters in the last permutation. Thus every binary $(r, 4)$-formation isomorphic to $I_{r}^{2} D_{r} I_{r}$ contains a copy of $u$.

Since every other binary $(r, 4)$-formation has a subsequence isomorphic to $I_{r} D_{r}^{2}$, then it suffices to observe that $I_{r} D_{r}^{2}$ contains a copy of $u$ that uses every letter in the third permutation. 
\end{proof}

\begin{corollary}\label{4.2} If $u$ is any sequence of the form $a v a v' a$ such that $a$ is a letter, $v$ is a nonempty sequence excluding $a$ with no repeated letters and $v'$ is obtained from $v$ by only moving the first letter of $v$ to another place in $v$, then $\mathit{Ex}(u, n)=\Theta(n\alpha(n))$.\end{corollary}

\begin{proof} The upper bound follows from the last lemma and Lemma \ref{1.4}, while the lower bound follows by Lemma \ref{1.1} since $u$ contains $a b a b a$.
\end{proof}

The next corollary is obtained by reversing the sequences considered in the last lemma.

\begin{corollary}\label{4.3} If $u$ is any sequence of the form $a v a v' a$ such that $a$ is a letter, $v$ is a nonempty sequence excluding $a$ with no repeated letters and $v'$ is obtained from $v$ by moving a single letter in $v$ to the end of $v$, then $\mathit{fw}(u)=4$ and $\mathit{Ex}(u, n)=\Theta(n\alpha(n))$.
\end{corollary}

The next lemma implies that if $v$ and $v'$ are nonempty permutations of the same distinct letters excluding $a$, then $\mathit{fw}(a v a v' a)=4$ if and only if $v'$ is obtained from $v$ by only moving the first letter of $v$ to another place in $v$ or by only moving a single letter in $v$ to the end of $v$.

\begin{lemma}\label{4.4} Let $u$ be any sequence of the form $a v a v' a$ such that $a$ is a letter, $v$ is a nonempty sequence excluding $a$ with no repeated letters, and $v'$ is a permutation of $v$ which cannot be obtained from $v$ by only moving the first letter of $v$ to another place in $v$ or by only moving a single letter in $v$ to the end of $v$. Then $\mathit{fw}(u)>4$.
\end{lemma}

\begin{proof} First note that $\mathit{fw}(x)>4$ if $x$ is $abcdadbca, abcdadcba, abcdeabdcea$, or $abcdeacbeda$. This can be verified using the formation width algorithm in the appendix. Suppose $u$ is a sequence of the form $0v0v'0$ for which $\mathit{fw}(u)=4$, $v$ is the sequence $12 \ldots r$, and $v'$ is the permutation $\pi_{1}\pi_{2}\ldots\pi_{r}$ of $12 \ldots r$. Since $u$ avoids $abcdadbca$ and $abcdadcba$, then $\pi_{i}\leq i+1$ for each $1\leq i\leq r$.

Consider two cases. In the first case, $\pi_{1}=1$. If $\pi_{i}=i$ for each $1\leq i\leq r$, then $\mathit{fw}(u)=4$ since $\mathit{fw}(\mathit{up}(r+1,2))=3$. Otherwise let $m$ be minimal for which $\pi_{m}=m+1$. Then $\pi_{j}=j$ for each $j<m$. Since $u$ avoids $abcdeabdcea$, then $\pi_{r}=m$. Moreover $\pi_{j}=j+1$ for $m\leq j<r$ since $\pi_{i}\leq i+1$ for each $1\leq i\leq r$. Thus $v'$ can be obtained from $v$ by only moving a single letter in $v$ to the end of $v$.

In the second case, $\pi_{1}=2$. Let $m$ be the index for which $\pi_{m}=1$. Then $\pi_{j}=j+1$ for $1\leq j<m$ since $\pi_{i}\leq i+1$ for each $1\leq i\leq r$. Since $u$ avoids $abcdeacbeda$, then $\pi_{j}=j$ for each $j>m$. Thus $v'$ can be obtained from $v$ by only moving the first letter of $v$ to another place in $v$.
\end{proof}

For $t\leq 4$ the next lemma exhibits sequences with three distinct letters and $t$ occurrences of each letter which contain $(a b)^{t}$ and have formation width $2t-1$.

\begin{lemma}\label{4.5} If $t$ is 2, 3, or 4 and $z$ is any sequence of the form $a x_{1} a x_{2}\ldots a x_{t}$ such that $a$ is a letter and $x_{i}$ is a sequence equal to either $bc$ or $cb$ for each $1\leq i\leq t$, then $\mathit{fw}(z)=2t-1$.
\end{lemma}

\begin{proof} The lower bound follows since $z$ contains $(a b)^{t}$. By Corollary \ref{2.2}, the upper bound is verified by checking that every binary $(3, 2t-1)$-formation contains $z$. The appendix has a program for running this check. 
\end{proof}

\begin{corollary}\label{4.6} If $t$ is $3$ or $4$ and $z$ is any sequence of the form $a x_{1} a x_{2}\ldots a x_{t}$ such that $a$ is a letter and $x_{i}$ is a sequence equal to either $bc$ or $cb$ for each $1\leq i\leq t$, then $\mathit{Ex}(z, n)=n2^{\frac{1}{(t-2)!}\alpha(n)^{t-2}\pm O(\alpha(n)^{t-3})}$.\end{corollary}

\begin{proof} The upper bounds follow from the last lemma and Lemma \ref{1.4}. The lower bounds follow from the lower bounds on $\mathit{Ex}((a b)^{t}, n)$ in \cite{2} by Lemma \ref{1.1}. \end{proof}

There are sequences $z$ of the form $a x_{1} a x_{2} a x_{3} a x_{4} a x_{5}$ such that $a$ is a letter and $x_{i}$ is a sequence equal to either $bc$ or $cb$ for each $1\leq i\leq 5$ for which $\mathit{fw}(z)>9$. For example $\mathit{fw}(abcacbacbabcacb) =10$.

The following lemma presents another infinite class of forbidden sequences with three distinct letters for which formation width yields tight bounds on extremal functions.

\begin{lemma}\label{4.7} $\mathit{fw}(abc(acb)^{t})=2t+1$ for $t\geq 0$.\end{lemma}

\begin{proof} The proof is trivial for $t=0$, so suppose that $t>0$. Since $abc(acb)^{t}$ contains an alternation of length $2t+2$, then $\mathit{fw}(abc(acb)^{t})\geq 2t+1$. In order to prove that $\mathit{fw}(abc(acb)^{t})\leq 2t+1$, it suffices by Corollary \ref{2.2} to show that every binary $(3,2t+1)$-formation contains $abc(acb)^{t}$. 

Consider any binary $(3, 2t+1)$-formation $f$ with permutations $xyz$ and $zyx$. Without loss of generality suppose that the last $2t-1$ permutations of $f$ have the subsequence $(xyz)^{t}$. Then $f$ has the subsequence $xzy(xyz)^{t}$ unless the first six letters of $f$ are $zyxxyz$. If the first six letters of $f$ are $zyxxyz$, then $f$ has the subsequence $zyx(zxy)^{t}$.
\end{proof}

\begin{corollary}\label{4.8} $\mathit{Ex}(abc(acb)^{t}, n)=n2^{\frac{1}{(t-1)!}\alpha(n)^{t-1}\pm O(\alpha(n)^{t-2})}$ for $t\geq 2$.\end{corollary}

\begin{proof} The upper bounds follow from the last lemma and Lemma \ref{1.4}. The lower bounds follow from the lower bounds on $\mathit{Ex}((a b)^{t}, n)$ in \cite{2} by Lemma \ref{1.1}.
\end{proof}

\section{Further bounds on $\mathit{fw}$}\label{sec5}

For $c\geq 2$ the bounds on $\mathit{l}(u)$ imply that $(c+1)k\leq \mathit{fw}(\mathit{alt}(c, 2k))\leq 2(c+1)k-1$ and $(c+1)k+1\leq \mathit{fw}(\mathit{alt}(c, 2k+1))\leq 2(c+1)k+1$ for every $k$. In this section we derive improved bounds on $\mathit{fw}(\mathit{alt}(c, 2k))$ and $\mathit{fw}(\mathit{alt}(c, 2k+1))$ using Corollary \ref{2.2}.

First we compute $\mathit{fw}(\mathit{alt}(c, 3))$ for all $c\geq 2$. Pettie showed in \cite{9} that $\mathit{Ex}(\mathit{alt}(c, 3), n)=O(n)$.

\begin{theorem}\label{5.1} If $c\geq 2$, then $\mathit{fw}(I_{c}D_{c}I_{c})=c+3$.\end{theorem}

\begin{proof} First we prove for every permutation $\pi\in S_{c}$ that $I_{\pi}D_{\pi}I_{\pi}$ is not a subsequence of the binary $(c, c+2)$-formation $I_{c}^{c}D_{c}^{2}$. Assume for contradiction that $I_{c}^{c}D_{c}^{2}$ has the subsequence $I_{\pi}D_{\pi}I_{\pi}$ for some permutation $\pi\in S_{c}$. Since $\mathit{l}(I_{c}D_{c})=c+1$ by Corollary \ref{3.2}, then the last letter of $D_{\pi}$ must be in the first $D_{c}$ in $I_{c}^{c}D_{c}^{2}$. However, the first letter of $I_{\pi}$ is the same as the last letter of $D_{\pi}$, so the first letter of the last $I_{\pi}$ in $I_{\pi}D_{\pi}I_{\pi}$ must be in the last $D_{c}$ in $I_{c}^{c}D_{c}^{2}$. Then $I_{\pi}=D_{c}$, so the last letter of $D_{\pi}$ is $c$. This would imply that $I_{c}^{c}c$ has the subsequence $I_{\pi}D_{\pi}$. Since the last letter of $I_{c}^{c}$ is $c$, then $I_{\pi}D_{\pi}$ would be a subsequence of $I_{c}^{c}$, a contradiction. Thus $I_{c}^{c}D_{c}^{2}$ does not have $I_{\pi}D_{\pi}I_{\pi}$ as a subsequence for any permutation $\pi\in S_{c}$. Thus $\mathit{fw}(I_{c}D_{c}I_{c}) >c+2$ by Corollary \ref{2.2}.

It remains to show that every binary $(c, c+3)$-formation $f$ has a subsequence $I_{\pi}D_{\pi}I_{\pi}$ for some permutation $\pi\in S_{c}$. Without loss of generality suppose the first permutation of $f$ is $I_{c}$. If $f$ is $I_{c}^{c+3}$, then $f$ has $I_{c}D_{c}I_{c}$ as a subsequence. If $f$ has an alternation of $I_{c}$ and $D_{c}$ terms of length at least 3, then also $f$ must have $I_{c}D_{c}I_{c}$ as a subsequence. Otherwise $f$ has the form $I_{c}^{a}D_{c}^{b}$ with $a+b=c+3$, $a>0$ and $b>0$. If $a \leq 2$, then $f$ has $I_{c}D_{c}I_{c}$ as a subsequence. If $b\leq 2$, then $f$ has $D_{c}I_{c}D_{c}$ as a subsequence. Otherwise $f$ has the subsequence $I_{\pi}D_{\pi}I_{\pi}$, such that $I_{\pi}$ is the sequence $I_{b-2}c\ldots (b-1)$ consisting of the integers from $1$ to $b-2$ followed by the integers in reverse from $c$ to $b-1$. In other words $I_{\pi}$ is obtained by reversing the last $a -1$ letters of $I_{c}$. Thus $\mathit{fw}(I_{c}D_{c}I_{c}) \leq c+3$ by Corollary \ref{2.2}. 
\end{proof}

The next two lemmas are used for the lower bounds in the remaining theorems.

\begin{lemma}\label{5.2} If $c\geq 2$ and $\pi\in S_{c}$, then $I_{\pi}D_{\pi}$ is a subsequence of $I_{c}^{c}D_{c}$ if and only if $\pi(1)<\pi(2)$. \end{lemma}

\begin{proof} Let $\pi\in S_{c}$ and suppose $I_{\pi}D_{\pi}$ is a subsequence of $I_{c}^{c}D_{c}$. Then the last letter of $I_{\pi}D_{\pi}$, namely $\pi(1)$, occurs in the last $D_{c}$ of $I_{c}^{c}D_{c}$ since $\mathit{l}(I_{c}D_{c})= c+1$ by Corollary \ref{3.2}. If $\pi(1)$ is not the only letter of $I_{\pi}D_{\pi}$ occurring in the last $D_{c}$, then $\pi(2)\pi(1)$ is a subsequence of $D_{c}$. This is possible only if $\pi(1)<\pi(2)$.

If the final $D_{c}$ contains no letters in $I_{\pi}D_{\pi}$ besides $\pi(1)$, then the last $\pi(2)$ in $I_{\pi}D_{\pi}$ occurs in some $I_{c}$. If $\pi(1)>\pi(2)$, then the last $\pi(1)$ in $I_{\pi}D_{\pi}$ can be replaced with the $\pi(1)$ in the same permutation as the last $\pi(2)$ in $I_{\pi}D_{\pi}$. This would imply that $I_{\pi}D_{\pi}$ is a subsequence of $I_{c}^{c}$, which is impossible since $\mathit{l}(I_{c}D_{c})=c+1$. Thus $\pi(1)<\pi(2)$.

For the other direction suppose that $\pi(1)<\pi(2)$. Then $I_{\pi}D_{\pi}$ is a subsequence of $I_{c}^{c+1}$ with exactly one letter in the last permutation of $I_{c}^{c+1}$. Thus $I_{\pi}D_{\pi}$ is a subsequence of $I_{c}^{c}D_{c}$.
\end{proof}

Define the reverse permutation $\pi_{r}\in S_{c}$ so that $\pi_{r}(i)=c+1-i$ for $1\leq i\leq c$.

\begin{corollary}\label{5.3} If $c\geq 2$ and $\pi\in S_{c}$, then $I_{\pi}D_{\pi}$ is a subsequence of $D_{c}I_{c}^{c}$ if and only if $\pi(2)<\pi(1)$.\end{corollary}

\begin{proof} By reflection, $I_{\pi}D_{\pi}$ is a subsequence of $D_{c}I_{c}^{c}$ if and only if $I_{\pi}D_{\pi}$ is a subsequence of $D_{c}^{c}I_{c}$. Moreover $I_{\pi}D_{\pi}$ is a subsequence of $D_{c}^{c}I_{c}$ if and only if $\pi_{r}(I_{\pi}D_{\pi})$ is a subsequence of $I_{c}^{c}D_{c}$. By Lemma \ref{5.2}, $\pi_{r}(I_{\pi}D_{\pi})$ is a subsequence of $I_{c}^{c}D_{c}$ if and only if $\pi_{r}(\pi(1))<\pi_{r}(\pi(2))$. Since $\pi_{r}(\pi(1))< \pi_{r}(\pi(2))$ if and only if $\pi(2)<\pi(1)$, then $I_{\pi}D_{\pi}$ is a subsequence of $D_{c}I_{c}^{c}$ if and only if $\pi(2)<\pi(1)$.
\end{proof}

Using these facts we determine $\mathit{fw}(I_{c}^{k}D_{c})$ and $\mathit{fw}(\mathit{alt}(c, 4))$. Pettie in \cite{9} showed bounds of $\Theta(n\alpha(n))$ on the maximum lengths of sequences with $n$ distinct letters avoiding both $ababab$ and $\mathit{alt}(c, 4)$ for some $c$. This improved an upper bound by Ezra, Aronov, and Sharir in \cite{4} on the complexity of the union of $n$ $\delta$-fat triangles.

\begin{theorem}\label{5.4} If $c\geq 2$ and $k\geq 1$, then $\mathit{fw}(I_{c}^{k}D_{c})=c+2k-1$.\end{theorem}

\begin{proof} The upper bound follows since $\mathit{fw}(I_{c}^{k}D_{c})\leq \mathit{fw}(I_{c}^{k})+c$.

For the lower bound let $T_{k}$ be the $(c, c+2k-2)$-formation obtained by concatenating $\mathit{alt}(c, 2k-2)$ and $I_{c}^{c}$. We show that $T_{k}$ avoids $I_{c}^{k}D_{c}$ by induction on $k$. This is clearly true for $k=1$ since $\mathit{l}(I_{c}D_{c})=c+1$ by Corollary \ref{3.2}. For the inductive hypothesis assume that $T_{k}$ avoids $I_{c}^{k}D_{c}$. Suppose for contradiction that $T_{k+1}$ has the subsequence $I_{\pi}^{k+1}D_{\pi}$ for some permutation $\pi\in S_{c}$.

The proof of Theorem \ref{3.5} showed that $\mathit{r}(I_{c}^{k}) = 2k-1$ and $I_{\pi}^{k}$ is a subsequence of $\mathit{alt}(c, 2k-1)$ only if $\pi(c)=c$, so the last $I_{\pi}D_{\pi}$ of $I_{\pi}^{k+1}D_{\pi}$ must be a subsequence of the rightmost $D_{c}I_{c}^{c}$ in $T_{k+1}$. Then $\pi(1)>\pi(2)$.

Since $T_{k}$ avoids $I_{c}^{k}D_{c}$, then the first letter $\pi(1)$ of the second $I_{\pi}$ in $I_{\pi}^{k+1}D_{\pi}$ must occur in the initial $I_{c}D_{c}$ of $T_{k+1}$. Thus $\pi(1)\pi(2)\pi(1)$ must be a subsequence of $I_{c}D_{c}$. This contradicts $\pi(1)>\pi(2)$, so $T_{k+1}$ avoids $I_{c}^{k+1}D_{c}$. Thus $\mathit{fw}(I_{c}^{k}D_{c})>c+2k-2$ for every $c\geq 2$ and $k\geq 1$ by Corollary \ref{2.2}. \end{proof}

\begin{theorem}\label{5.5} If $c\geq 2$, then $\mathit{fw}(I_{c}D_{c}I_{c}D_{c})=2c+3$. \end{theorem}

\begin{proof} Since $c+\mathit{fw}(I_{c}D_{c}I_{c}) \geq \mathit{fw}(I_{c}D_{c}I_{c}D_{c})$, then $2c+3\geq \mathit{fw}(I_{c}D_{c}I_{c}D_{c})$. As for the lower bound, the $(c, 2c+2)$-formation $F=I_{c}^{c}D_{c}^{2}I_{c}^{c}$ avoids $I_{\pi}D_{\pi}I_{\pi}D_{\pi}$ for all permutations $\pi\in S_{c}$. To see this assume for contradiction that $F$ contains $I_{\pi}D_{\pi}I_{\pi}D_{\pi}$ for some permutation $\pi\in S_{c}$. Since $I_{c}^{c}$ does not contain $I_{\pi}D_{\pi}$ by Corollary \ref{3.2}, then the first $I_{\pi}D_{\pi}$ is in the first $I_{c}^{c}D_{c}$ of $F$ and the second $I_{\pi}D_{\pi}$ is in the last $D_{c}I_{c}^{c}$ of $F$. This is a contradiction by Lemma \ref{5.2} and Corollary \ref{5.3}. Thus $\mathit{fw}(I_{c}D_{c}I_{c}D_{c}) >2c+2$ by Corollary \ref{2.2}. 
\end{proof}

We extend the technique used in the last proof to get an improved lower bound on $\mathit{fw}(\mathit{alt}(c, k))$ for all $c\geq 2$ and $k\geq 5$.

\begin{theorem}\label{5.6} If $c\geq 2$ and $k\geq 1$, then $\mathit{fw}(\mathit{alt}(c, 2k))\geq k(c+2)-1$ and $\mathit{fw}(\mathit{alt}(c, 2k+1))\geq k(c+2)+1$.\end{theorem}

\begin{proof} Define $T_{1}=I_{c}^{c}, T_{2k}=T_{2k-1}D_{c}^{2}$, and $T_{2k+1}=T_{2k}I_{c}^{c}$ for $k\geq 1$. We prove that $T_{k-1}$ avoids $\mathit{alt}(c, k)$ by induction on $k$. This implies that $\mathit{fw}(\mathit{alt}(c, 2k))>k(c+2)-2$ and $\mathit{fw}(\mathit{alt}(c, 2k+1))>k(c+2)$ by Corollary \ref{2.2}. Theorems \ref{5.1} and \ref{5.5} proved that $T_{2}$ avoids $\mathit{alt}(c, 3)$ and $T_{3}$ avoids $\mathit{alt}(c, 4)$.

For the inductive hypothesis there are two cases. First assume that $T_{j-1}$ avoids $\mathit{alt}(c, j)$ for all $j\leq 2k-1$, but suppose for contradiction that $T_{2k-1}$ has the subsequence $(I_{\pi}D_{\pi})^{k}$ for some permutation $\pi\in S_{c}$. Let $G$ be the leftmost $(I_{\pi}D_{\pi})^{k-1}$ in the subsequence $(I_{\pi}D_{\pi})^{k}$. Since the leftmost $T_{2k-3}$ in $T_{2k-1}$ avoids $\mathit{alt}(c, 2k-2)$, then the last letter of $G$ must occur somewhere in the rightmost $D_{c}^{2}I_{c}^{c}$ in $T_{2k-1}$. Moreover the letter directly after $G$ in $(I_{\pi}D_{\pi})^{k}$ is the same as the last letter of $G$, so these two letters cannot occur in the same permutation. Thus the last $I_{\pi}D_{\pi}$ in $(I_{\pi}D_{\pi})^{k}$ must be a subsequence of the last $D_{c}I_{c}^{c}$ in $T_{2k-1}$. Then $\pi(2)<\pi(1)$ by Corollary \ref{5.3}.

Let $H$ be the rightmost $(I_{\pi}D_{\pi})^{k-1}$ in the subsequence $(I_{\pi}D_{\pi})^{k}$. Since the rightmost $T_{2k-3}$ in $T_{2k-1}$ avoids $\mathit{alt}(c, 2k-2)$, then the first letter of $H$ must occur somewhere in the leftmost $I_{c}^{c}D_{c}^{2}$ in $T_{2k-1}$. Moreover the letter directly before $H$ in $(I_{\pi}D_{\pi})^{k}$ is the same as the first letter of $H$, so these two letters cannot occur in the same permutation. Thus the first $I_{\pi}D_{\pi}$ in $(I_{\pi}D_{\pi})^{k}$ must be a subsequence of the first $I_{c}^{c}D_{c}$ in $T_{2k-1}$. Then $\pi(1)<\pi(2)$ by Lemma \ref{5.2}, a contradiction.

For the second case of the inductive hypothesis, assume that $T_{j-1}$ avoids $\mathit{alt}(c, j)$ for all $j\leq 2k$, but suppose for contradiction that $T_{2k}$ has the subsequence $(I_{\pi}D_{\pi})^{k}I_{\pi}$ for some permutation $\pi\in S_{c}$. Let $G$ be the leftmost $(I_{\pi}D_{\pi})^{k}$ in the subsequence $(I_{\pi}D_{\pi})^{k}I_{\pi}$. Since the leftmost $T_{2k-1}$ in $T_{2k}$ avoids $G$, then the last letter of $G$ must occur in the last $D_{c}^{2}$ in $T_{2k}$. The last letter of $G$ is equal to the first letter of the last permutation of $(I_{\pi}D_{\pi})^{k}I_{\pi}$, so the last $I_{\pi}$ of $(I_{\pi}D_{\pi})^{k}I_{\pi}$ must be a subsequence of the final $D_{c}$ in $T_{2k}$. Therefore $I_{\pi}=D_{c}$, so the last letter of $D_{\pi}$ is $c$. This implies that $(I_{\pi}D_{\pi})^{k}$ is a subsequence of $T_{2k-1}c$, so $(I_{\pi}D_{\pi})^{k}$ would be a subsequence of $T_{2k-1}$, a contradiction. Thus $(I_{\pi}D_{\pi})^{k}I_{\pi}$ is not a subsequence of $T_{2k}$ for any permutation $\pi\in S_{c}$.
\end{proof}

\section{Open Problems}\label{sec6}

Many questions about formation width are left unresolved by the results in this paper. We found several classes of sequences $u$ for which $u$ contained an alternation with the same formation width as $u$, which implied tight bounds on $\mathit{Ex}(u, n)$. One problem is to find all sequences $u$ for which $u$ contains an alternation with the same formation width as $u$.

We showed that $\mathit{fw}(abc(acb)^{t})=2t+1$ for $t\geq 0$, which implied that $\mathit{Ex}(abc(acb)^{t}, n) =$ $n2^{\frac{1}{(t-1)!}\alpha(n)^{t-1}\pm O(\alpha(n)^{t-2})}$ for $t\geq 2$. We conjecture the following result, which would imply nearly tight bounds on $\mathit{Ex}(abc(acb)^{t}abc, n)$.

\begin{conjecture}\label{6.1} $\mathit{fw}(abc(acb)^{t}abc)=2t+3$ for $t\geq 0$.\end{conjecture}

We identified the set of all sequences $u$ for which $\mathit{fw}(u)\leq 3$. These are all the sequences for which the value of $\mathit{fw}(u)$ implies linear bounds on $\mathit{Ex}(u, n)$. A next step would be to identify all sequences $u$ for which $\mathit{fw}(u)\leq 4$, since these are all of the sequences for which the value of $\mathit{fw}(u)$ implies $O(n\alpha(n))$ upper bounds on $\mathit{Ex}(u, n)$.

We also determined the values of $\mathit{l}(u)$ and $\mathit{r}(u)$ for every binary formation $u$. Since both of these functions provide lower bounds on $\mathit{fw}(u)$, it would be useful to compute the values of $\mathit{l}(u)$ and $\mathit{r}(u)$ for every sequence $u$.

On a related note, the values of $\mathit{l}(u)$ implied bounds on $\mathit{fw}(u)$ within a factor of $2$ for every binary formation $u$. What is the exact value of $\mathit{fw}(u)$ for every binary formation $u$?

We also obtained bounds on $\mathit{fw}(\mathit{alt}(c, k))$ for every $k\geq 1$. In particular we determined the exact values for $k\leq 4$. What is the exact value of $\mathit{fw}(\mathit{alt}(c, k))$ for each $k\geq 5$?

In addition we proved that $\mathit{fl}(u)\leq(r-1)^{2^{\mathit{fw}(u)-1}}+1$ for all sequences $u$ with $r$ distinct letters. What is the exact value of $\mathit{fl}(u)$ for every sequence $u$?


\subsection*{Acknowledgments}
We thank the referee for many helpful suggestions, as well as Peter Tian, Tanya Khovanova, Pavel Etingof, Slava Gerovitch, Katherine Bian, Jacob Fox, and Peter Shor. We also thank the MIT PRIMES program for supporting this research.


\appendix
\section{Algorithm for computing $\mathit{fw}$}

The following algorithm for computing $\mathit{fw}(u)$ is an implementation in Python of the method for computing formation width in Corollary \ref{2.2}. Specifically if $u$ is a nonempty sequence with $r$ distinct letters, then the algorithm increments $s$ starting from $1$ until it finds that every binary $(r,s)$-formation contains $u$. The longest common subsequence functions are from the post by MarkF6 at http://stackoverflow.com/questions/10746282/longest-common-subsequence-of-three-strings.

\begin{verbatim}
import string 
from collections import defaultdict
from itertools import permutations
  
#computes longest common subsequence: 

def lcs_grid(xs, ys): 
    grid = defaultdict(lambda: defaultdict(lambda: (0,""))) 
    for i,x in enumerate(xs): 
        for j,y in enumerate(ys): 
            if x == y: 
                grid[i][j] = (grid[i-1][j-1][0]+1,"\\") 
            else: 
                if grid[i-1][j][0] > grid[i][j-1][0]: 
                    grid[i][j] = (grid[i-1][j][0],"<") 
                else: 
                    grid[i][j] = (grid[i][j-1][0],"^") 
    return grid 
  
def lcs2(xs,ys): 
    grid = lcs_grid(xs,ys) 
    i, j = len(xs) - 1, len(ys) - 1
    best = [] 
    length,move = grid[i][j] 
    while length: 
        if move == "\\": 
            best.append(xs[i]) 
            i -= 1
            j -= 1
        elif move == "^": 
            j -= 1
        elif move == "<": 
            i -= 1
        length,move = grid[i][j] 
    best.reverse() 
    return best 
  
#determines whether one sequence is a subsequence of another: 

def issubseq(seq, subseq): 
    if len(lcs2(seq, subseq)) == len(subseq): 
        return True
    else: 
        return False
  
#constructs set of binary (l, s)-formations: 
  
def rsform(l,s): 
    rsformset = set() 
    if s == 0: 
        return rsformset 
    rsformset1 = set() 
    q = tuple(range(l)) 
    q1 = q[::-1] 
    rsformset.add(q) 
    for i in range(s-1): 
        for rsform in rsformset: 
            t = rsform+q 
            rsformset1.add(t) 
            t = rsform+q1 
            rsformset1.add(t) 
        rsformset.clear() 
        for rsform in rsformset1: 
            rsformset.add(rsform) 
        rsformset1.clear() 
    return rsformset 
  
#determines the formation width of u:

def formwidth(u, l):
    count = 0
    s=1
    v = list(u) 
    while True:
        count = 0
        for rsforms in rsform(l, s):
            for perms in permutations(range(l)):
                for i in range(len(u)):
                    v[i] = perms[u[i]]
                if issubseq(rsforms, v):
                    count = count+1
                    break
        if count == len(rsform(l, s)):
            return s
        else:
            s = s+1

def fw(u):
    t = set(u)
    return formwidth(u, len(t))

#u must be nonempty tuple with letters 0,1,2,..., e.g.:
print fw((0,1,2,3,4,5,0,2,3,1,4,5,0))
\end{verbatim}

\end{document}